\documentclass[review]{elsarticle}

\usepackage{lineno,hyperref}
\modulolinenumbers[5]

\journal{Journal of \LaTeX\ Templates}

\usepackage{graphicx}
\usepackage{url}
\usepackage{bbm}
\usepackage{caption}
\usepackage{subcaption}
\usepackage{amsfonts}
\usepackage{amssymb}
\usepackage{kantlipsum}
\usepackage{enumerate}
\usepackage{amsthm}
\usepackage[T1]{fontenc}
\usepackage{multirow}
\usepackage{amsmath}
\usepackage{bbm}
\usepackage{dsfont}
\usepackage{xcolor}
\usepackage{booktabs}
\usepackage{mathtools}
\usepackage[utf8]{inputenc}
\usepackage{enumerate}
\usepackage{algorithm}
\usepackage{algpseudocode}
\newcommand{\remove}[1]{}
\newtheorem{theorem}{Theorem}[section]
\newtheorem{lemma}{Lemma}[section]
\newtheorem{remark}{Remark}[section]

\algdef{SE}[DOWHILE]{Do}{doWhile}{\algorithmicdo}[1]{\algorithmicwhile\ #1}
\def\E{{\mathbb{E}}}
\def\R{{\mathbb{R}}}

\DeclareMathOperator*{\argmin}{arg\,min}
\algnewcommand{\algorithmicgoto}{\textbf{go to}}%
\algnewcommand{\Goto}[1]{\algorithmicgoto~\ref{#1}}

\begin{document}
%%%%%%%%%%%%%%%%%%%%%%%%%%%%%%%%%

\begin{frontmatter}

\title{Service Scheduling for Random Requests with Quadratic Waiting Costs}

\author{Ramya Burra, Chandramani Singh and Joy Kuri}% <-this % stops a space
%\thanks{The authors are with the Department
%of Electronic Systems Engineering, Indian Institute of Science Bangalore, 560012 India. Emails: {\{{burra,chandra,kuri}\}@iisc.ac.in} Telephone: +91 80 22932495}% <-this % stops a space
%\thanks{The first and second authors acknowledge support from Research Fellowships of Visvesvaraya PhD Scheme and INSPIRE Faculty Research Grant (DSTO-1363).}}
\address{Department of ESE, Indian Institute of Science, Bangalore}



\cortext[mycorrespondingauthor]{Corresponding author}
%\ead{\{{burra,chandra,kuri}\}@iisc.ac.in}
\ead{burra,chandra,kuri@iisc.ac.in}
\begin{abstract}

%\color{blue}
We study service scheduling problems in a slotted system in which agents arrive with service requests according to a Bernoulli process and have to leave within two slots after arrival, service costs are quadratic in service rates, and there are also waiting costs. We consider quadratic waiting costs. We frame the problems as average cost Markov decision processes. While the studied system is a linear system with quadratic costs, it has state dependent control. Moreover, it also possesses a non-standard cost function structure in the case of fixed waiting costs, rendering the optimization problem complex. We characterize optimal policy. We provide an explicit expression showing that the optimal policy is linear in the system state. We also consider systems in which the agents make scheduling decisions for their respective service requests keeping their own cost in view. We consider quadratic waiting costs and frame this scheduling problems as stochastic games. We provide Nash equilibria of this game. To address the issue of unknown system parameters, we propose an algorithm to estimate them. We also bound the cost difference of the actual cost incurred and the cost incurred using estimated parameters. 
%\color{black}
\end{abstract}
\begin{keyword}
Service Scheduling, Quadratic waiting cost, Markov Decision Process
\end{keyword}

\end{frontmatter}
\section{Introduction}
In several systems, agents are admitted at slot boundaries, but they can
leave as soon as their services are complete, e.g., consider EVs at EV Charging stations. Then the waiting period of an agent can depend on the amount of the deferred service. It is reasonable to consider waiting costs that depend on the amount of the deferred service in such cases. In~\cite{EV_waiting_cost}, the authors introduce a non-decreasing convex penalty on EVs' average waiting time. %We have considered linear waiting costs in Chapter~\ref{Chapter:Linear}. %The authors in~\cite{quadratic_EV_dissatisfaction} consider systems where the service facility need not deliver complete service leading to dissatisfaction of agents which they capture using a dissatisfaction cost - a quadratic function of the unfinished amount of service. In a similar context, the authors in~\cite{LangTong_Restless_MAB} consider a convex dissatisfaction cost. 
Quadratic waiting costs capture users' higher sensitivity to incremental delays while still rendering the problems in the class of linear systems with quadratic costs. In this work, we consider quadratic waiting costs and analyze the resulting scheduling problems. In particular, we consider the cases where the jobs can stay for two slots but incur a quadratic waiting cost in second slots. We see that this service scheduling problem
is a special case of constrained linear quadratic control. We study both optimal scheduling and Nash equilibria in case of selfish agents. We analyze optimal and equilibrium policies for this problem.

\subsection{Related work}
%In case of charging EVs optimal scheduling minimizes the total charging cost. 
In~\cite{Bae-Kwasinski-12}, the authors propose a centralized algorithm to minimize the total charging cost of EVs. It determines the optimal amount of charging to be received at various charging stations en route. There is another line of work which intends to minimize waiting times at the charging stations. For instance, in~\cite{Gusrialdi-et-al-14}  the authors propose a distributed scheduling algorithm that uses local information of traffic flows measured at the neighbouring charging stations to uniformly utilize charging resources along the highway and minimize the total waiting time. In our work, we consider minimizing both charging and waiting costs simultaneously. More precisely we look at quadratic waiting costs.
In the context of traffic routing and scheduling, the authors in~\cite{b1} consider a scenario where agents compete for a common link to ship their demands to a destination. They
obtain the optimal and equilibrium flows in the presence of polynomial congestion cost. 

In~\cite{b10}, we consider routing on a ring network in the presence of quadratic congestion costs and also linear delay costs when traffic is redirected through the adjacent nodes.  However, the problems in~\cite{b10} are one-shot optimization problems as these do not have a temporal component.

Scheduling for minimizing energy costs has also been considered
in the context of CPU power consumption~\cite{b2,b11}, big data processing~\cite{b3},
production scheduling in plants~\cite{b7}. In~\cite{b26}, the authors propose an optimal online algorithm for job arrivals with deadline uncertainty. In this work, they consider convex processing cost. They also derive competitive ratio for the proposed algorithm. None of these studies accounts for waiting costs of jobs as considered in our work.

In~\cite{our-journal}, we studied service scheduling for Bernoulli job arrivals, quadratic service costs and linear waiting costs. We obtained a piece-wise linear optimal policy. We also studied Nash equilibrium in this setting.

\subsection{Our Contribution}
%\subsubsection*{Quadratic Waiting Costs:}
\begin{enumerate}
\item We study optimal scheduling in the presence of quadratic waiting costs. While this problem fits in the standard framework of linear quadratic control Markov decision problems, however, it does not meet certain controllability requirements. Here we derive the optimal scheduling policy for the case where jobs' service
requirements are identical.
\item We also provide an algorithm that yields the optimal control for general service requirements.
\item We obtain a symmetric Nash equilibrium for the associated stochastic game. 
\end{enumerate}
We also present a comparative numerical study to illustrate the impact of quadratic waiting cost structure and performance criteria~(optimal scheduling vs strategic scheduling by selfish agents). 

List of our contributions can be found in the Table~\ref{table:contributions}.

\begin{table}[htb]
%\footnotesize
\caption{List of contributions}
\label{table:contributions}
   \centering
%\resizebox{\columnwidth}{!}{%
%\begin{tabular}{ |p{3cm}||p{4cm}|p{4cm}|  }
\scriptsize
 \begin{tabular}{|c||c|}
 %\hline
%\multicolumn{3}{|c|}{List of Contribution} \\
 \hline
Versions     & Policy \\
 \hline
 Optimal scheduling (Bernoulli arrivals)      & Exact policy (Section~\ref{SectionQuadratic:quad-waiting-cost})\\

 \hline
  Optimal Scheduling (General arrivals) & Exact policy (Section~\ref{QuadracticSection:general})\\

 \hline
 Nash Equilibrium (Bernoulli arrivals)   & Exact policy (Section~\ref{SectionQuadratic:nash-equilibrium})\\
  
 \hline
\end{tabular}
%}
\end{table}
\section{System Model}
\label{SectionQuadratic:system-model}
We consider a time-slotted system where time is divided into discrete slots. Service requests arrive over slots to the service facility. Each request has to be completely served before its deadline. The deadline of a job is fixed at $2$ slots after its arrival. So service can be scheduled such that portions of the requests are served in the future slots before their respective deadlines. Serving requests incur a cost, and the price in a slot depends on the quantum of service delivered in that slot. We consider two scheduling problems: one where the service provider makes scheduling decisions in order to optimize the overall time-average cost and the other where the agents who bring the jobs make scheduling decisions for their respective jobs to minimize their individual costs. Below we present the system model and both the problems formally.
\subsection{Service request model} Agents with service requests arrive according to an i.i.d. Bernoulli$(p)$ process;~$ p \in (0,1)$. All the agents demand $\psi$ amount of service. Further, each request can be met in at most two slots, i.e., a fraction of
the demand arriving in a slot could be deferred to the next slot.
\subsection{Cost model} \label{SectionQuadratic:system-model-fixed-cost}The cost consists of two components:
\begin{itemize}
\item {\it Service cost:} The service price in a slot is a linear function of the total service
offered in that slot. Thus the total service cost in a slot is square
of the total offered service in that slot.
\item {\it Waiting cost:} We consider a scenario where a request's waiting cost is a quadratic function of the portion of service that is deferred. Each request incurs a waiting cost $dx^2$ where $x$ is the portion of its demand deferred to the next slot.   %We also study quadratic waiting costs which are described in appropriate sections later.
\end{itemize}
We consider the following two scheduling problems.
\subsection{Performance Criteria}
\subsubsection{Optimal Scheduling}
We aim to minimize the time-averaged cost of the service provider. Let, for $k \geq 1$,  $x_k$ be the remaining demand from slot $k-1$ to
slot $k$; $x_1 = 0$. This demand must be met in slot $k$.
Also, for $k \geq 1$, let $v_k$ be the extra service offered in slot $k$.
Clearly, $v_k \in [0,\psi]$ and is $0$ if there is no request in slot $k$.
A {\em scheduling policy} $\overline{\pi} = (\pi_k,k \geq 1)$
is a sequence of functions $\pi_k:[0,\psi] \rightarrow [0,\psi]$
such that if there is a service request in slot $k$ then $\pi_k(x_k)$ gives the amount of service deferred from slot $k$ to slot $k+1$. 
%Let $x_k,v_k$ and $\pi_k(x)$ be as defined in Section~\ref{SectionLinear:system-model}. 
More precisely, we want to determine the scheduling policy $\overline{\pi}$
that minimizes
\begin{equation}
\lim_{T \to \infty} \frac{1}{T}\sum_{k=1}^T\E[(x_k + v_k)^2 + d {x_{k}^2}].
\label{Quadraticeqn:average-cost}
\end{equation}
We obtain the optimal solution in Section ~\ref{SectionQuadratic:quad-waiting-cost}.

\subsubsection{Equilibrium for Selfish Agents}
Setup is similar to~\cite[Section~II B]{our-journal}. However, the expected cost of an agents is different as the waiting cost in this work is quadratic waiting costs. The expected cost of an agent who arrives in slot $k$, if it sees a remaining demand $x$,
is
\begin{align}
c_k(x,\overline{\pi}) = (\psi-\pi_k(x))(\psi-\pi_k(x)+x) + \pi_k(x)(\pi_k(x)+ p (\psi-\pi_{k+1}(\pi_k(x))) )+d \pi_k^2(x).\label{Quadraticeqn:selfish-agent-cost-expression}
\end{align}
We focus on symmetric Nash equilibria of the form $(\pi,\pi,\dots)$
and obtain one such equilibrium in Section~\ref{SectionQuadratic:nash-equilibrium}.
%\color{black}
\section{Optimal Scheduling}
\label{SectionQuadratic:quad-waiting-cost}
We first show that the optimal scheduling problem can
be transformed into a stochastic shortest path problem.
Towards that, from the {\it Renewal Reward Theorem}~\cite{renewal-reward-theorem}, and~\cite[Lemma~3.1]{our-journal} the following holds
\begin{align*}
\lim_{T \to \infty} \frac{1}{T}\sum_{k=1}^T\E[(x_k + v_k)^2 + d x_k^2   = p(1-p)\E\left[\sum_{k = A_i}^{A_{i+1}-1}\left((x_k + v_k)^2 + d x_k^2\right)\right].
\end{align*}
We now frame the problem as {\it stochastic shortest path problem} where terminal state corresponds to
absence of request in a slot similar to~\cite{our-journal}.

\paragraph*{Stochastic shortest path formulation}
We let $x_k$ and $u_k$ denote the remaining demand from slot $k-1$ to slot $k$
and the service offered in slot $k$, respectively. In particular, we let $x_k$ be the system state in slot $k$ and $t$ be the {\it terminal state}
which is hit if there is no new request in a slot.
Let $u_k$ be the action in slot $k$ provided $x_k$ is 
not a terminal state; $u_k \in [x_k, x_k+\psi]$. 
Given the state-action pair in slot $k$, $(x_k,u_k)$, the next state is $x_{k+1} = x_k + \psi - u_k$ with probability $p$ and the terminal state with probability $1-p$. The single stage cost before hitting the terminal state is $u_k^2 + d x_k^2$  and the terminal cost is $x_{k+1}^2(1+d)$. 

Unlike linear waiting cost problems, we can cast the unconstrained problem as a standard linear quadratic
control Markov decision problem. Towards this, let us redefine the system state at slot $k$~(if it is not the terminal state) to be 
\[y_k \coloneq \begin{bmatrix} 
x_k & \psi 
\end{bmatrix}^T.\] Clearly, the states evolve as 
\[
y_{k+1} = \begin{cases}
A y_k+Bu_k & \text{if slot $k+1$ has a request}, \\
t, & \text{otherwise},
\end{cases}
\]
where
 \[A=\begin{bmatrix} 
1 & 1 \\
0 & 1 
\end{bmatrix} \text{ and }B=\begin{bmatrix} 
-1 \\
0 
\end{bmatrix}.
\]
The single stage cost and the terminal cost can be written as $y_k^TQy_k+u_k^TRu_k$ and $y_{k+1}^THy_{k+1}$, respectively, where
\[Q=\begin{bmatrix} 
d & 0 \\
0 & 0 
\end{bmatrix}, 
R=1 
\text{ and } H=\begin{bmatrix} 
d+1 & 0 \\
0 & 0 
\end{bmatrix}
\]
Note that $Q$ and $H$ are positive semi-definite matrices whereas $R$ is positive definite as required in the standard framework of linear quadratic control problems~(see~\cite[Section~3.2]{b20}).\footnote{The framework in~\cite[Section~3.2]{b20} require that the system state evolve as $y_{k+1} = Ay_k + Bu_k + w_k$ where independent random vectors with zero mean and finite second moments. Moreover, $w_k$s must also be independent of $y_k$s and $u_k$s.  In our setup, the system evolves in deterministic fashion until it hits the terminal state. In particular, $w_k = 0$ for all $k$ until $y_{k+1} = t$. Hence the above requirement is met.}

Standard framework~\cite[Section~3.2]{b20}, requires the pairs $(A,B)$ and $(A,C)$, where $Q = C^TC$, are controllable and observable, respectively
~(see also~\cite[Proposition~4.1]{b21}). For readability, we provide the definitions of ``controllable" and ``observable" in the following.

\paragraph*{Definition}A pair $(A,B)$, where $A$ is an $n \times n$ matrix and $B$ is an $n \times m$ matrix, is said to be controllable if the $n \times nm$ matrix
\[
[b, AB,A^2b,\dots,A^{n-1}B]
\]
has full rank. A pair $(A,C)$ , where $A$ is an $n \times n$ matrix and $C$ is an $m \times n$ matrix, is said to be observable if the pair $(A^T,C^T)$ is controllable, where $A^T,C^T$ denote the transposes of $A$ and $C$, respectively. 

We can easily verify that $(A,C)$ is observable but $(A,B)$ is not controllable in our setup. Below, we explicitly obtain the optimal policy.

Let $J:[0, \psi] \rightarrow \R_+$ be the optimal cost function~(see~\cite[Chapter~1]{b21}, for definition of optimal cost function) for the problem. It is the solution of the following Bellman's equation: For all $x \in [0,\psi]$,
\begin{align}
\label{Quadraticeqn:quad-delay-cost-bellman}
J(x)=\min_{u \in [0,\psi]}\left\{(x+\psi-u)^2+dx^2+pJ(u)
+(1-p)u^2(1+d)\right\}
\end{align}
Let $\pi^\ast$ be the optimal stationary policy for this problem. Let us define the "$k$-stage problem" and
\remove{
as the one that allows at most $k+1$  service requests. More precisely, here the system is {\it forced to enter}
 the terminal state after $k+1$ service requests if it has not already done so. }
 let $J_k(\cdot)$ be the optimal cost function
 of the $k$-stage problem.Clearly,
 \begin{equation}
J_0(x)=  \min_{u \in [0,\psi]}\left\{(\psi-u+x)^2+dx^2+u^2(1+d)\right\}
\label{Quadraticeqn:J0}
 \end{equation}
 and
  \begin{align}
 J_k(x)= \min_{u \in [0,\psi]} 
\left\{(\psi-u+x)^2+dx^2+pJ_{k-1}(u)
 +(1-p)u^2(1+d) \right\}  \label{Quadraticeqn:Jk}.
 \end{align}
 for $k \geq 1$. We can express $J(\cdot)$ as the limit
of $J_k(\cdot)$ as $k$ approaches infinity.
Furthermore, we can express the desired optimal
policy also as the limit of the optimal controls
of $k$-stage problems~(i.e., optimal actions in~\eqref{Quadraticeqn:J0}-\eqref{Quadraticeqn:Jk}). This is the approach we follow to arrive at the
optimal scheduling policy.

\subsection{Optimal Policy}
Let us define sequences $a^{\ast}_i,b^{\ast}_i, i \geq 0$
as follows.
\begin{align}
a^\ast_i &= \begin{cases}
1+d, & \text{if } i = 0, \\
1+d-\frac{p}{1+a^{\ast}_{i-1}}, & \text{otherwise,}
\end{cases} \label{Quadraticeqn:ak-star} \\
b^\ast_i &= \begin{cases}
0, & \text{if } i = 0, \\
\frac{p(2a^{\ast}_{i-1}\psi+b^{\ast}_{i-1})}{1+a^{\ast}_{i-1}} & \text{otherwise.}
\end{cases} \label{Quadraticeqn:bk-star}
\end{align}
We first state a few properties of the above sequences.
\begin{lemma} $(a)$ The sequence $a^\ast_k, k\geq 0$ is a decreasing sequence and   converges to $a_\infty := \frac{d+\sqrt{d^2+4(1+d-p)}}{2}$.\\
$(b)$ The sequence $b^\ast_k, k\geq 0$ converges to
\[
b_{\infty} := \frac{2pa_\infty\psi}{1+a_\infty-p}.
\]
Further, $b^\ast_k < 2\psi$ for all $k \geq 0$ and so, $b_{\infty} \le 2\psi$. \\
\label{Quadraticlemma:monotonicity-ak-bk}
\end{lemma}
\begin{proof}
See Appendix~\ref{Quadraticapp:monotonicity-ak-bk}.
\end{proof}

\begin{lemma}
 $0 < \frac{x+\psi-\frac{b_{i}}{2}}{(1+a_{i})} < \psi$ for all $0 \leq x \leq \psi,i \ge 0$.
\label{Quadraticlemma:monotonicity-xk}
\end{lemma}
\begin{proof}
See Appendix~\ref{Quadraticapp:monotonicity-xk}.
\end{proof}

The optimal scheduling policy is as follows.
\begin{theorem} 
%\label{theorem:single-psi}
\[
\pi^\ast(x) = \frac{x+\psi-\frac{b_{\infty}}{2}}{(1+a_{\infty})}.
\]
\label{Quadraticthm:optimal-policy}
\end{theorem}
\begin{proof}
See Appendix~\ref{Quadraticapp:optimal-policy}.
\end{proof}
The optimal policies here are linear. When the pending service in a slot is $x$ and $u$ amount of service is deferred, the marginal service cost in the slot is lower bounded by $2(\psi-u+x)$ and the marginal waiting cost is upper bounded by $2du$.  Hence irrespective of the values of $x$, it is profitable to defer some amounts of service to the next slot.%This can be understood as follows. In the case of quadratic waiting costs, incremental waiting costs start at zero irrespective of the value of $d$ and increase with the amount of deferred service. On the other hand, marginal service costs are positive for any nonzero amount of offered service in a slot. Hence, for any given pending service and $d$, the optimal policies defer nonzero amounts of service.    

We illustrate the optimal policies via a few examples in Figure~\ref{Quadraticfig:Quadratic-image1}. We choose $\psi = 2, d= 1,$ and $p = 0.5, 0.85$ and $1$ for illustration. As expected, for the same pending service, the deferred service decreases as p increases. For $p =1$, there is no pending service in the first slot and no amount of service is deferred in the subsequent slots either.
%\color{black}
\begin{figure}[!ht]
	
	\centering
	
	\includegraphics[width=0.5\textwidth]{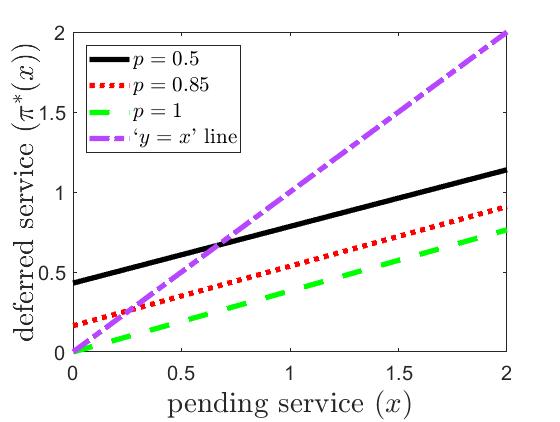}
	\caption{The optimal policies for $\psi=2,d=1,p=\{0.5,0.85,1\}$.}
	\label{Quadraticfig:Quadratic-image1}
\end{figure}
\section{Optimal Scheduling for General Service Requirements}
\label{QuadracticSection:general}
We now generalize the service request process of Section~\ref{SectionQuadratic:quad-waiting-cost} to allow general service requirements. We assume that, in each slot an agent with demand $\psi_i$~($i = 1,2,\dots,N$) arrives with probability $p_i$ and there is no arrival with probability $1 - \bar{p}$ where $\bar{p} \coloneqq \sum_{i=1}^N p_i$. Without loss of generality we assume that $\psi_i$s are monotonically increasing.  

Let us see the stochastic shortest path formulation of this problem. Let $J:\{\psi_1,\dots,\psi_N\} \times [0,\psi_N] \to \mathbb{R}_+$ be the optimal cost function and $\pi:\{\psi_1,\dots,\psi_N\} \times [0,\psi_N] \to [0,\psi_N]$  be the optimal policy for the problem~($\pi(\psi_i,\cdot) \in [0,\psi_i]$ for all $i$). The optimal cost function is solution of the following Bellman's equation: For all $x \in [0,\psi_N], i \in \{1,2,\dots,N\}$,
 \begin{align*}
 J(\psi_i,x)= \min_{u \in [0, \psi_i]} 
\Bigg\{(\psi_i-u+x)^2 & +dx^2 + \sum_{j=1}^N p_j J(\psi_j,u)\\
+ (1-\bar{p}) u^2(1+d)\Bigg\}
\end{align*}
Using a procedure similar to~\cite[Section~V-A]{our-journal} we propose Algorithm~\ref{alg:two-service-requirements} which provides the optimal policy. The policy derived after $k$ runs of the {\em do-while} loop is the optimal policy, $\pi_k(\psi_i,\cdot)~(i = 1,2,\dots,N)$ of an appropriately defined $k$-stage problem. We see that the termination criterion of the loop is met after a few iterations in most of the cases. In other words, $\pi_k(\cdot,\cdot), k \geq 0$ converge to $\pi(\cdot,\cdot)$ in  a few iterations. Unlike the case of Bernoulli arrivals in Section~\ref{SectionQuadratic:quad-waiting-cost}, the optimal policies here can be piecewise linear though they do not exhibit discontinuities.
\remove{
\begin{remark}
    We can propose an approximate policy on the lines of Section~\ref{SectionLinear:Optimal-scheduling-fictitious}. We can easily check that the same approximation error bounds as in Theorem~\ref{Linearthm:bounding-error} apply here as well.
\end{remark}
}
We illustrate the optimal policies for general service requirements via a few examples in Figure~\ref{Quadraticfig:Quadratic-image6}. We choose $(\psi_1, \psi_2) = (1,3), d= 1,$ and  $(p_1,p_2) = (0.2,0.7)$ and $(0.7,0.2)$ for illustration. As expected, more service is deferred when load in the current slot is higher, and so, $\pi(\psi_1,\cdot) \leq \pi(\psi_2,\cdot)$. For both the $(p_1,p_2)$ combinations, $x^2_{k,0} < 0$, and so $\pi(\psi_2,0) > 0$. $\pi(\psi_1,\cdot)$ are capped at $\psi_1$. Moreover, for the same pending service, the deferred service decreases as the expected load in the next slot increases, i.e., for given $x$ and $i = 1,2$, $\pi(\psi_i,x)$ for $(p_1,p_2) = (0.2,0.7)$ are smaller than  $\pi(\psi_i,x)$ for $(p_1,p_2) = (0.7,0.2)$.

%\color{black}

\begin{figure}[!ht]
	
	\centering
	
	\includegraphics[width=0.5\textwidth]{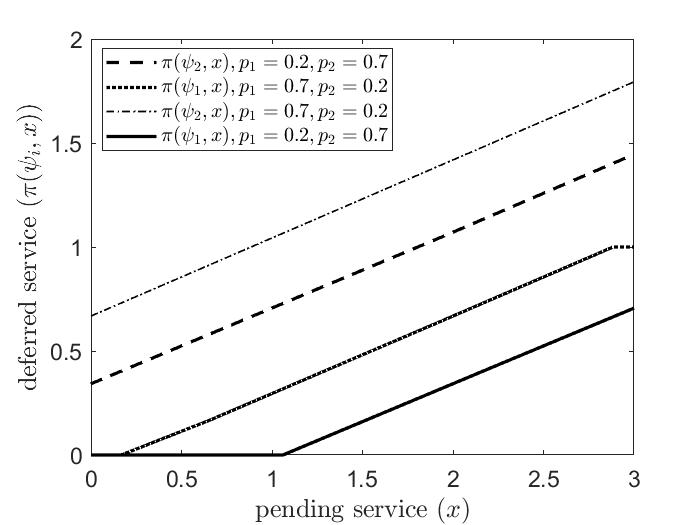}
	\caption{Optimal policies for $(\psi_1,\psi_2) = (1,3), d=1, (p_1,p_2)\in \{(0.2,0.7),(0.7,0.2)\}$.}
	\label{Quadraticfig:Quadratic-image6}
\end{figure}
\remove{

As in Section~\ref{SectionQuadratic:system-model}, we assume that agents with service requests arrive according to an i.i.d. process. However, in each slot, with probability $p_i$ an agent with demand $\psi_i$ arrives, and with probability $1-\sum_{i=1}^N p_i$ there is no arrival. Further, each request can be met in at most two slots, i.e., a fraction of the demand arriving in a slot could be deferred to the next slot. We consider that $\psi_i$s are in ascending order without loss of generality. We assume that the cost structure is similar to that in Section~\ref{SectionQuadratic:system-model}. }
\begin{algorithm}
\caption{(General Service Requirements)}
\label{alg:two-service-requirements}
\begin{algorithmic}
\State Input: $p_1,p_2,\dots,p_N,\psi_1,\psi_2,\dots,\psi_N,d$
\State $a_{k,-1}=\infty,b_{k,-1}=0~\forall k \ge 0$
\State $k=0$
\State $x_{0,0}=0,x_{0,1}=\psi_N,I_0=1$
\State $a_{0,0}=1+d,b_{0,0}=0$
\Do
\State $k=k+1$
\For{$i=1:N$}
\For{$j=0:I_{k-1}-1$}
\State \[x_{k,j}^i=\frac{2(1+a_{k-1,j})x_{k-1,j}+b_{k-1,j}}{2}-\psi_i\]
\EndFor
\EndFor
\For{$i=1:N-1$}
\vspace{-0.25in}
\State \begin{align*}&\bar{l}(i)=\max\{j:x_{k-1,j}< \psi_i\} \\
&x_{k,\bar{l}(i)+1}^1=\frac{2(1+a_{k-1,\bar{l}(i)})\psi_i+b_{k-1,\bar{l}(i)}}{2}-\psi_i
\end{align*}
\vspace{-0.25in}
\EndFor
\begin{align*}
(x_{k,0},\dots,x_{k,I_{k}})=&
\text{order}(x^1_{k,0},\dots,x^1_{k,\bar{l}(1)+1}
,\dots,x^{N-1}_{k,0},\dots,x^{N-1}_{k,\bar{l}(N-1)+1},x^N_{k,0},\dots,x^N_{k,I_{k-1}-1},0,\psi_2)
\end{align*}
\indent\Comment{This function removes the values outside $[0,\psi_N]$ and puts the remaining in ascending order. }
\For{$j=0:I_{k}-1$}
\For{$i=1:N$}
\begin{equation*}
    j_i=
    \begin{cases}
    -1,&\text{if}\ x_{k,0}^i>x_{k,j}\\
    \max\{l:x^i_{k,l}\le x_{k,j}\},&\text{otherwise}
    \end{cases}
\end{equation*}
\EndFor
\vspace{-0.4in}
\State \begin{align*}
a_{k,j}=&1-\sum_{m=1}^{N-1}\frac{p_m}{1+a_{k-1,j_m}}\mathds{1}_{\{j_m\le \bar{l}(m)\}}-\frac{p_N}{1+a_{k-1,j_N}}\\
b_{k,j}=&\sum_{m=1}^{N-1}\frac{p_m(2\psi_ma_{k-1,j_m}+b_{k-1,j_m})}{1+a_{k-1,j_m}}\mathds{1}_{\{j_m\le \bar{l}(m)\}}+\frac{p_N(2\psi_Na_{k-1,j_N}+b_{k-1,j_N})}{1+a_{k-1,j_N}}+d
\end{align*}
\EndFor 
%\State $n=n+1$
\doWhile{$(x_k,a_k,b_k)\neq (x_{k-1},a_{k-1},b_{k-1})$}
\State Output: $\forall i \in \{1,2,\dots,N\}$
\begin{equation*}
\pi(\psi_i,x) =
%\label{equ:b1-expressions-2psi}
\begin{cases}

0, &\text{if}\ x \le x_{k,0}^i\\

\big[\frac{2(x+\psi_i)-b_{k,j}}{2(1+a_{k,j})}\big]^{\psi_i}, &\text{if} \  x\in (x_{k,j}^i,x_{k,j+1}^i]\\
&0 \leq j < I_k\\
\end{cases}
  \end{equation*}
\end{algorithmic}
\end{algorithm}
\section{Nash equilibrium}
\label{SectionQuadratic:nash-equilibrium}
In this section we provide a Nash equilibrium for the non-cooperative game among the selfish agents~(see Section~\ref{SectionQuadratic:system-model}). As in~\cite{our-journal}, we focus on symmetric Nash equilibria where each agent's strategy is a piece-wise linear function of the remaining demand of the previous player. Our notation for agents' strategies and costs and analysis closely follow those in Section IV. Now the optimal cost of a player as a function of the pending demand given that all other players use strategy, $\pi'(\cdot)$ is given by
\remove{

%We omit the proofs for lack of space.
Let $C:[0,\psi]\to \mathbb{R}_+$ give the optimal cost for a player as a function of the pending demand given that all other players use strategy $\pi':[0,\psi]\to[0,\psi]$. Clearly, $C(x)$ is given by the following equation for all $x\in[0,\psi]$.} 
\[
C(x)=\min_{u \in [0,\psi]}\{(\psi-u)(\psi-u+x)+u(u+p(\psi-\pi'(u)))+du^2\}
\]
Also, $\bar{\pi}'=(\pi',\pi',\dots)$ a symmetric nash equilibrium 
\remove{
if ${\pi}'(x)$ attains the optimal cost in the above optimization problem for all $x $, i.e.,}if 
\[
\pi'(x)\in \argmin_{u \in [0,\psi]}\{(\psi-u)(\psi-u+x)+u(u+p(\psi-\pi'(u)))+du^2\},
\]
for all $x \in [0,\psi]$. We characterize one such Nash equilibrium in the following. We define $k$-stage problems as in~\cite{our-journal}.
\subsubsection*{A symmetric Nash equilibrium} Let us define sequences $a'_k,b'_k,k \ge -1$ as follows
\begin{align}
a'_k &= \begin{cases}
0, & \text{if } k = -1 \\
\frac{1}{2(2+d-pa'_{k-1})}, & \text{otherwise}
\end{cases} \label{Quadraticeqn:a-case2-nash} \\
b'_k &= \begin{cases}
0, & \text{if } k = -1 \\
\frac{(2-p)\psi+pb'_{k-1}}{2(2+d-pa'_{k-1})}, & \text{otherwise}
\end{cases} \label{Quadraticeqn:b-case2-nash}
\end{align}
We first state a few properties of the above sequences.
\begin{lemma}
\label{Quadraticlemma:nash-properties-a-b}
%\begin{enumerate}
    $(a)$ The sequence $a'_k, k \ge -1$ converges to \[a'_{\infty}\coloneqq \frac{4+2d}{4p}-\frac{\sqrt{(2+d)^2-2p}}{2p}.\] Also,  $a'_{\infty}<\frac{1+\frac{d}{2}}{p}$.\\
   $(b)$ The sequence $b'_k, k \ge -1$ converges to \[b'_{\infty}\coloneqq \frac{a'_{\infty}(2-p)\psi}{1-a'_{\infty}p}.\]
%\end{enumerate}
\end{lemma}
\begin{proof}
See Appendix~\ref{Quadraticappendix:nash-properties-a-b}.
\end{proof}
\begin{lemma}
 $0 < a'_k x+b'_k < \psi$ for all $0 \leq x \leq \psi,k \ge 0$.
\label{Quadraticlemma:quad-game-no-caps}
\end{lemma}
\begin{proof}
See Appendix~\ref{Quadraticapp:quad-game-no-caps}.
\end{proof}
\begin{theorem}
\label{Quadraticthm:quad-nash-equilibrium}
$\bar{\pi}'=(\pi',\pi',\dots)$ is a symmetric Nash equilibrium where\\
\[
\pi'(x)=a'_{\infty}x+b'_{\infty}, \forall x \in [0,\psi].
\]
\end{theorem}
\begin{proof}
See Appendix~\ref{Quadraticappendix:quad-nash-equilibrium}.
\end{proof}
Observe that, similar to the optimal policies in Section~\ref{SectionQuadratic:quad-waiting-cost}, the symmetric Nash equilibria given by the above theorems are also linear. 

We now illustrate symmetric Nash equilibria for the same parameters as used to illustrate the optimal polices in Section~\ref{SectionQuadratic:quad-waiting-cost} in Figure~\ref{Quadraticfig:Quadratic-image2}. As in the optimal policies, for the same pending service, the deferred service decreases as p increases. For $p =1$, the system attains a steady state wherein each user observes a pending service $0.5231$ (the fixed point of $\pi'(x) = x$ in Figure~\ref{Quadraticfig:Quadratic-image2}) and defers the same amount of service. Consequently, the amount of offered service in each slot equals $\psi$ in the steady state.
\begin{figure}[!ht]
	\centering
	\includegraphics[width=0.5\textwidth]{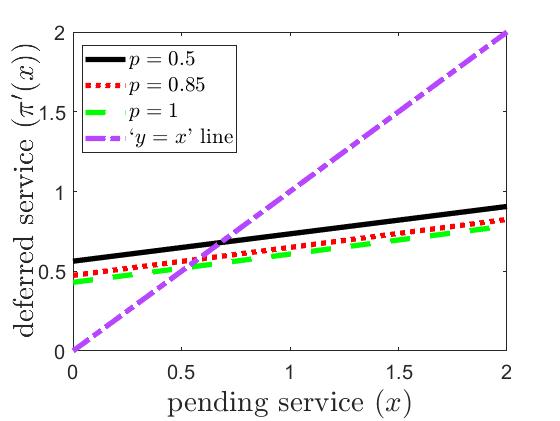}
	\caption{The Nash equilibrium policies for $\psi=2,d=1,p=\{0.5,0.85,1\}$.}
	\label{Quadraticfig:Quadratic-image2}
\end{figure}
\color{black}
\section{Unknown system parameters}
All throughout this work we assumed that arrival statistics are known to the service facility. However, in many real time applications it may not be available to the service facility. To deal with such scenarios one has to learn the unknown parameter on the go. The action at any slot should be guided by the current estimate of the parameter in that slot. However, this process is a cumbersome process. So as a first step towards this, we first estimate the parameter upto $\epsilon$ accuracy with high probability. Then, we propose to use this estimate for deciding on action in any slot. In the following we first outline the details of estimating the unknown parameter. Subsequently, we provide an upper bound on the difference of the cost incurred when the parameter is known and the cost incurred when the parameter is unknown.
\subsection{Estimating the unknown parameter $p$}
\label{Quadraticsec:unknown}
Let us define a sequence of random variables $\{X_i\}_{i \ge 1}$. If there is an arrival in slot $k$, then $X_k=1$ else $X_k=0$. Note that $X_i$s are independent random variables bounded in $[0,1]$. 
\paragraph*{Algorithm} In any slot $k \ge 1$, the estimate of parameter $p$ is $\frac{\sum_{i=1}^{k} X_i}{k}$. 

More precisely, we are counting frequency of arrivals. Using Hoeffding's inequality, the following holds
\begin{equation*}
    P\Big(|\frac{1}{n}\sum_{i=1}^nX_i-p|\ge \epsilon \Big) \le 2e^{-2\epsilon^2n}.
\end{equation*}
The following is the quantity of our interest, from the above inequality the following holds
\begin{equation*}
    P\Big(|\frac{1}{n}\sum_{i=1}^nX_i-p|\le \epsilon \Big) \ge 1-2e^{-2\epsilon^2n}.
\end{equation*}
Let $h$ be the desired high probability for the estimate to be $\epsilon$ close, then from the above inequality after $\tilde{n}=-\frac{1}{2\epsilon^2}\log{\frac{1-h}{2}}$ slots, the estimate and the original parameter $p$ are $\epsilon$ close with atleast probability $h$. Precisely, for fixed $\epsilon,h$, there exists a $\tilde{n}=-\frac{1}{2\epsilon^2}\log{\frac{1-h}{2}}$  such that $\forall n \ge \tilde{n}$
\begin{equation*}
    P\Big(|\frac{1}{n}\sum_{i=1}^nX_i-p|\le \epsilon \Big) \ge h.
\end{equation*}
Note that $\tilde{n}$ is a function of $\epsilon,h$.
\subsection{Bound on the cost difference}
Recollect that in the context of Bernoulli arrivals the optimal costs can be defined as follows from~\eqref{Quadraticeqn:J0} and \eqref{Quadraticeqn:Jk}
 \begin{equation}
J_0(x)=  \left\{(\psi-\pi^\ast(x)+x)^2+dx^2+\pi^\ast(x)^2(1+d)\right\}
\label{Quadraticeqn:J0-pi}
 \end{equation}
 and $\forall k \ge 1$
  \begin{align}
 J_k(x)= \left\{(\psi-\pi^\ast(x)+x)^2+dx^2+pJ_{k-1}(\pi^\ast(x))
 +(1-p)\pi^\ast(x)^2(1+d) \right\}  \label{Quadraticeqn:Jk-pi},
 \end{align}
where $\pi^\ast(x)$ is the optimal scheduling policy as defined in Theorem~\ref{Quadraticthm:optimal-policy}.

We now consider a setup where the parameter $p$ is unknown to the service facility. We estimate the parameter using algorithm in Section~\ref{Quadraticsec:unknown} till $\tilde{n}$ slots. This assures that the estimate and the original parameter are $\epsilon$ with probability $h$ following the arguments in Section~\ref{Quadraticsec:unknown}. Let us call this estimate to be $\tilde{p}$. Therefore, we use the following policy for a pending service of $x$ units.
 \begin{align*}
\tilde{\pi}(x) &= \frac{x+\psi-\frac{b_{\infty}(\tilde{p})}{2}}{(1+a_{\infty}(\tilde{p}))},\\
\text{ where }a_\infty(\tilde{p}) &= \frac{d+\sqrt{d^2+4(1+d-\tilde{p})}}{2}\\
\text{and }b_{\infty}(\tilde{p}) &= \frac{2\tilde{p}a_\infty(\tilde{p}) \psi}{1+a_\infty(\tilde{p})-\tilde{p}}.
\end{align*}
Note that the above policy is same as the optimal policy, however as we are unaware of the original $p$, we use the estimated $\tilde{p}$ instead of that. We now consider the cost of this system starting after $\forall n \ge \tilde{n}$ samples, it can be defined as follows
 \begin{equation*}
 \tilde{J}(x)= \left\{(\psi-\tilde{\pi}(x)+x)^2+dx^2+p\tilde{J}(\tilde{\pi}(x))
 +(1-p)\tilde{\pi}(x)^2(1+d) \right\}.
 \end{equation*}
 Note that for any fixed $\epsilon,h$ the cost function $\tilde{J}(x)$ depends on the estimate $\tilde{p}$, which in turn depends on random variables $X_i,i \le \tilde{n}$. Hence, $\tilde{J}(x)$ is a random variable.
 The $k$ stage problem can be defined as follows 
 \begin{equation}
\tilde{J}_0(x)=  \left\{(\psi-\tilde{\pi}(x)+x)^2+dx^2+\tilde{\pi}(x)^2(1+d)\right\}
\label{Quadraticeqn:J0-pi-tilda}
 \end{equation}
 and $\forall k \ge 1$
  \begin{align}
 \tilde{J}_k(x)= \left\{(\psi-\tilde{\pi}(x)+x)^2+dx^2+p\tilde{J}_{k-1}(\tilde{\pi}(x))
 +(1-p)\tilde{\pi}(x)^2(1+d) \right\}  \label{Quadraticeqn:Jk-pi-tilda}.
 \end{align}
For any fixed $\epsilon,h$ the cost function $\tilde{J}_k(x),\forall k \ge 0$ is also a random variable as they depend on the random variables $X_i,i \le \tilde{n}$. In the following, we would like to bound $\tilde{J} - J$. Notice that $\lim_{k \to \infty}\tilde{J}_k(x) = \tilde{J}(x)$ for all $x \in [0, \psi]$. We derive bound for $\tilde{J}_k(x)-J_k(x)$ for all $k \ge 1$. We then take $k \to \infty$ to obtain a bound on $\tilde{J}(x)-J(x)$. To derive this bound we first need the following lemma
\begin{lemma}
\label{Lemma:Quadratic-newLemma}
$|\pi^\ast(x)-\tilde{\pi}(x)| \le K |\tilde{p}-p|$, almost surely where 
\[
K=\Bigg[  \frac{8\psi}{(2+d+\sqrt{d^2+4d})^2\sqrt{d^2+4d}}+ 4\psi \frac{4 (2+d) \Bigg(   \frac{(1+d)(2+d)+\frac{1}{\sqrt{d^2+4d}}}{(d+\sqrt{d^2+4d})^2} \Bigg) + \Bigg(\frac{1}{\sqrt{d^2+4d}}\Bigg)}{{(d+\sqrt{d^2+4d})^2}} \Bigg]
\]
\end{lemma}
\begin{proof}
See Appendix~\ref{Appendix:newLemma}
\end{proof}
\begin{lemma}
\label{lemma:result-a}
For all $x_1,x_2 \in [0,\psi]$ the following holds
\[
|\tilde{\pi}(x_1)-\pi^\ast(x_2)| \le K.|\tilde{p}-p|+z|x_1-x_2|,
\]
almost surely,where $z=\frac{1}{1+d}$.
\end{lemma}
\begin{proof}
See Appendix~\ref{Appendix:result-a}
\end{proof}
Let us define the following notation $\forall~n \in \{0,1,,\dots\}$ which is required to proceed further.
\begin{align*}
   \tilde{\pi}^n(x)&=\tilde{\pi}\underbrace{(\tilde{\pi}(\dots\tilde{\pi}(x))\dots)}_{n \text{ times}}\\
   {\pi^\ast}^n(x)&={\pi^\ast}\underbrace{({\pi^\ast}(\dots{\pi^\ast}(x))\dots)}_{n \text{ times}}
\end{align*}
Note that by the above definition the following holds $\forall~n \in \{0,1,\dots\}$
\begin{align*}
    \tilde{\pi}^n(x)&=\tilde{\pi}(\tilde{\pi}^{n-1}(x))\\
    {\pi^\ast}^n(x)&={\pi^\ast}({\pi^\ast}^{n-1}(x))
\end{align*}
The following lemma bounds $|\tilde{\pi}(\tilde{\pi}^{n-1}(x))-\pi^\ast({\pi^\ast}^{n-1}(x))|$. This bound plays a crucial role in deriving on the bound on the cost differences.
\begin{lemma}
\label{lemma:result-b}
For all $n \ge 0$,
\[
|\tilde{\pi}(\tilde{\pi}^{n-1}(x))-\pi^\ast({\pi^\ast}^{n-1}(x))|=|\tilde{\pi}^n(x)-{\pi^\ast}^n(x)|\le \sum_{i=0}^n z^i K|\tilde{p}-p|
\]
almost surely.
\end{lemma}
\begin{proof}
See Appendix~\ref{appendix:result-b}
\end{proof}
The following lemma bounds the cost difference of $0$ stage problem.
\begin{lemma}
\label{lemma:result-c}
For all $y \in \{0,1,2,\dots\}$,
\[
\tilde{J}_0(\tilde{\pi}^y(x))-J_0({\pi^\ast}^y(x)) \le K'|\tilde{p}-p|,
\]
almost surely, where $K'=\frac{2\psi K (5+2d)}{(1-z)}$.
\end{lemma}
\begin{proof}
See Appendix~\ref{Appendix:result-c}
\end{proof}
\begin{lemma}
\label{lemma:result-II}
For all $n\in\{1,2,..,k-1\}$, almost surely
\[
\tilde{J}_{k-n}(\tilde{\pi}^{n-1}(x))-{J}_{k-n}({\pi^\ast}^{n-1}(x)) \le \Bigg[(A+B)\sum_{j=0}^{k-n}p^j \sum_{i=0}^{n-1+j}z^i+Az^n \sum_{j=0}^{k-n} (pz)^j+p^{k-n}K'\Bigg]|\tilde{p}-p|, 
\]
where
\begin{align*}
    A &= 2K\psi (2+(1-p)(1+d)), \text{ and }\\
    B &= 2K\psi(2+d)\\
   \end{align*}
\end{lemma}
\begin{proof}
See Appendix~\ref{Appendix:result-II}
\end{proof}

The major theorem that bounds the cost difference is as follows
\begin{theorem}
\label{theorem:bound-unknown}
For fixed $\epsilon,h$, there exists a $\tilde{n}=-\frac{1}{2\epsilon^2}\log{\frac{1-h}{2}}$  such that $\forall n \ge \tilde{n}$, with a probability of atleast $h$ the following holds
    \[
    \tilde{J}(x)-J(x) \le \Bigg[A+\frac{A+B}{(1-z)(1-\tilde{p}-\epsilon)}+\frac{Az}{1-(\tilde{p}+\epsilon)z}\Bigg]\epsilon
    \]
\end{theorem}
\begin{proof}
See Appendix~\ref{Appendix:bound-unknown}
\end{proof}

\color{black}

\section{Numerical Evaluation}
We now discuss the effect of the waiting cost structure on the scheduling policies, deferred services and costs. We also compare the impact of performance criteria~(optimal scheduling vs strategic scheduling by selfish agents).   

Let us revisit the optimal policies and Nash equilibria in Figures~\ref{Quadraticfig:Quadratic-image1} and~\ref{Quadraticfig:Quadratic-image2}. Recall that we had chosen $\psi = 2, d= 1,$ and  $p = 0.5, 0.85$ and $1$.   Notice that for the same parameters, amount of deferred service under the optimal policy is more sensitive to pending service than amount of deferred service under the Nash equilibrium. The equilibria are not as sensitive to $p$ as the optimal policies. 
We show histograms of pending services seen by the jobs for both optimal policies and Nash equilibria in Figure~\ref{Quadraticfig:Quadratic-image3,4}. We use $p  = 0.5$ and $p = 0.85$ for upper and lower subfigures respectively. In both the plots, $(1-p)$ fraction of jobs see $y_0 = 0$ pending service, and for $k \geq 1$, $p^k (1-p)$ fraction of jobs see $y_k = \pi(y_{k-1})$ pending service~($\pi \equiv \pi^\ast$ for an optimal policy whereas $\pi \equiv \pi'$ for a Nash equilibrium). For all $k \geq 0$, $y_k$ are upper bounded by the fixed point of $\pi(x) = x$. For $p = 0.85$, under Nash equilibrium the system attains a
steady state wherein each user observes a pending service $= 0.62$ (the fixed point of $\pi'(x)=x$
in Figure~\ref{Quadraticfig:Quadratic-image2} and defers the same amount of service). Hence we see a big mass at the fixed point of  $\pi'(x) = x$. Following the same reason we see a mass at the fixed point of $\pi^\ast(x) = x$. 

Finally, in Figure~\ref{Quadraticfig:Quadratic-image5}, we show variation of time-average cost under both optimal policy and Nash equilibrium as $p$ is varied from $0$ to $1$. We consider two sets of other parameters, $\psi = 2, d =1$ and $\psi = 2.5, d = 1.5$. For  $p = 1$, no service is deferred in any slot under the optimal policy, and hence the optimal average cost is $\psi^2$. For $\psi = 2, d =1$ and $p =1$, under the Nash equilibrium, $\psi$ service is offered and $0.5232$ service is deferred in each slot, and hence the average cost is $2^2 + 0.5232^2$. The efficiency loss is $1$ for $p \gtrsim 0$ and $1 + \frac{0.5232^2}{2^2}=1.0684$ for $p =1$. We make similar observations for $\psi =2.5, d = 1.5$. 
\remove{
%\begin{subfigures}
%	\centering
\begin{figure}[!ht]	
\centering
	\includegraphics[width=0.5\textwidth]{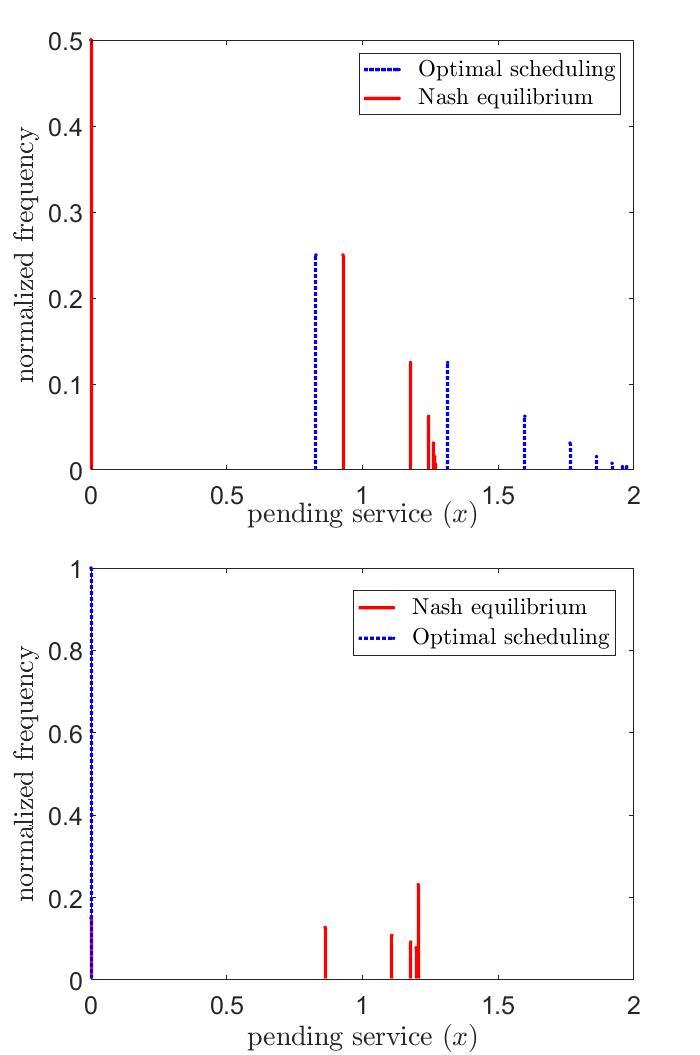}
	\caption{Fixed waiting costs: histogram of the pending services seen by the jobs for $\psi=2,d=1,p=0.5$(top subfigure) and $p=0.85$(bottom subfigure).}
	\label{Quadraticfig:fixed-image3,4}
\end{figure}

\begin{figure}[!ht]
	
	\centering
	
	\includegraphics[width=0.5\textwidth]{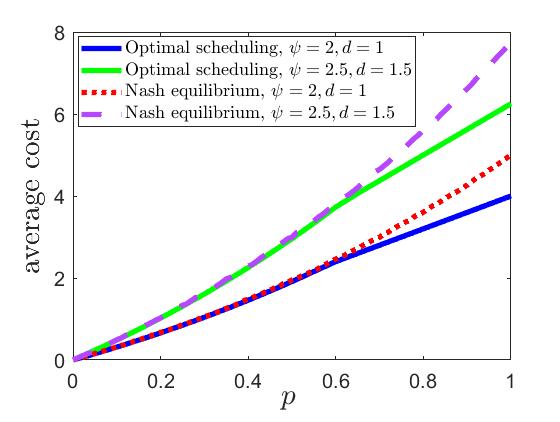}
	\caption{Fixed waiting costs: average cost vs. $p$ for $\psi=2,d=1$ and $\psi=2.5,d=1.5$.}
	\label{Quadraticfig:fixed-image5}
\end{figure}
%\begin{subfigures}
}
\begin{figure}[!ht]
	
	\centering
	
	\includegraphics[width=0.5\textwidth]{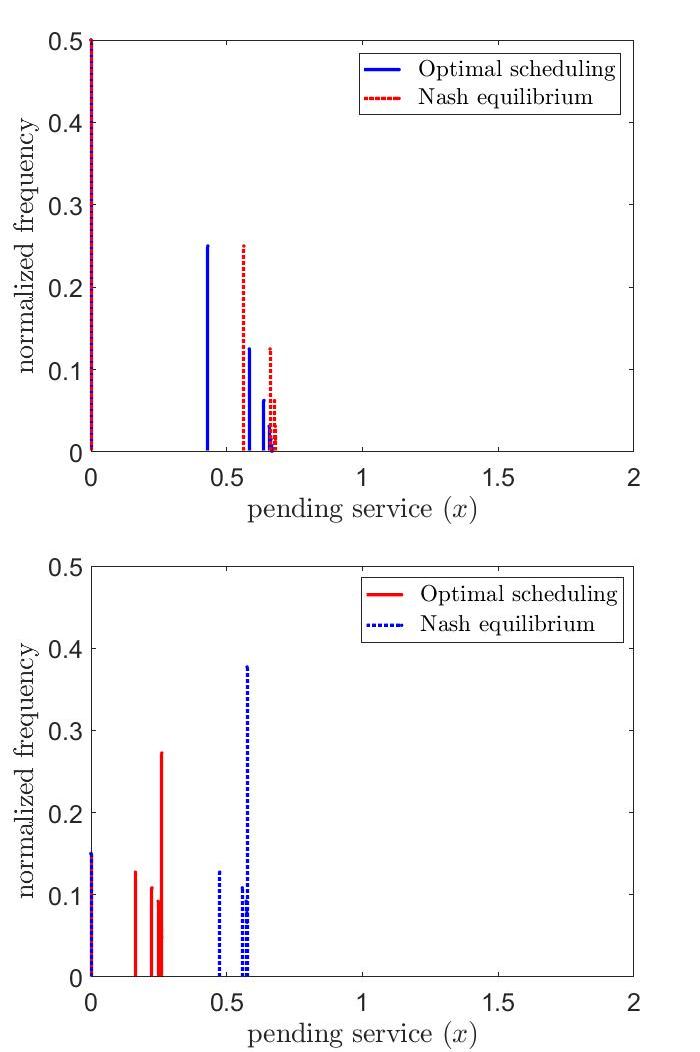}
	\caption{Quadratic waiting costs: histogram of the pending services seen by the jobs for $\psi=2,d=1,p=0.5$(top subfigure) and $p=0.85$(bottom subfigure).}
	\label{Quadraticfig:Quadratic-image3,4}
\end{figure}
%\begin{figure}[!ht]
	
%	\centering
	
%	\includegraphics[width=0.3\textwidth]{Qimage4.jpg}
%	\caption{Quadratic waiting costs: histogram of the pending services seen by the jobs for $\psi=2,d=1,p=0.85$.}
%	\label{Quadraticfig:Quadratic-image4}
%\end{figure}
%\end{subfigures}
\begin{figure}[!ht]
	
	\centering
	
	\includegraphics[width=0.5\textwidth]{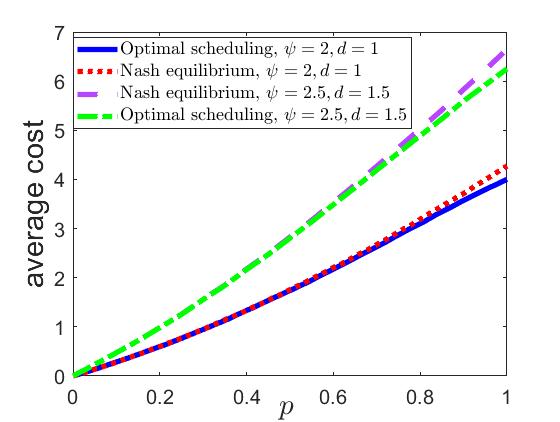}
	\caption{Quadratic waiting costs: average cost vs. $p$ for $\psi=2,d=1$ and $\psi=2.5,d=1.5$.}
	\label{Quadraticfig:Quadratic-image5}
\end{figure}
\section{Conclusion}
We studied service scheduling in slotted systems with Bernoulli request arrivals, quadratic service costs, quadratic waiting costs and
service delay guarantee of two slots.
Initially we study the case of jobs with identical service
requirements and we provided explicit optimal policy~(Theorem~\ref{Quadraticthm:optimal-policy}). We also gave the algorithm to compute the optimal policy if request could have different service requirements~(Algorithm~\ref{alg:two-service-requirements}). For competing requests, with identical
service requirements, we derived a symmetric Nash equilibrium~(Theorem~\ref{Quadraticthm:quad-nash-equilibrium}). To address the issue of unknown system parameters, we propose an algorithm to estimate them. We also bound the cost difference of the actual cost incurred and the cost incurred using estimated parameters~(Theorem~\ref{theorem:bound-unknown}).

%\subsubsection*{Acknowledgments:}
\bibliography{references}

\remove{

\bibliography{main}
}
%\bibliographystyle{ieeetr}
%\bibliography{references}
%\appendices
%\def\thesubsectiondis{\Roman{subsection}.}
%\newtheorem{apptheorem}{Theorem}[subsection]
%\newtheorem{applemma}[apptheorem]{Lemma}
%\bibliographystyle{IEEEtran}
%\bibliography{references}

%\appendices
%\def\thesectiondis{\Roman{section}}
%\newtheorem{apptheorem}{Theorem}[subsection]
%\newtheorem{applemma}[apptheorem]{Lemma}
\begin{appendix}
\subsection{Proof of Lemma~\ref{Quadraticlemma:monotonicity-ak-bk}}
\label{Quadraticapp:monotonicity-ak-bk}
$(a)$ Notice the mapping $a \mapsto 1+d-\frac{p}{1+a}$ is monotonically increasing.
Further, $a^\ast_0 = 1+d$ and $a^\ast_1 = 1+d-\frac{p}{2+d} < a^\ast_0$.
Therefore the sequence $a^\ast_k, k\geq 0$ is monotonically decreasing.
It is also non negative, and so, lower bounded. There are two solutions to the fixed point of $a = 1+d-\frac{p}{1+a}$ are as follows.
\[
\frac{d+\sqrt{(d+2)^2-4p}}{2},\frac{d-\sqrt{(d+2)^2-4p}}{2}.
\]
As $p \ge 0,d>0$ the following holds
\[
\frac{d-\sqrt{(d+2)^2-4p}}{2}<\frac{d+\sqrt{(d+2)^2-4p}}{2}\le 1+d
\]
Hence it converges to $a_\infty$, the largest fixed point
of $a = 1+d-\frac{p}{1+a}$.\\
$(b)$ We first show that $b^\ast_i, i \geq 0$ are bounded. Towards this, observe that
$b^\ast_i \leq {pb^\ast_{i-1} }+ {2p(1+d)\psi}$ for all $i \geq 1$. In particular,
$b^\ast_1 \leq {pb^\ast_{0} }+ {2p(1+d)\psi},~b^\ast_2 \leq {pb^\ast_{1} }+ {2p(1+d)\psi} \leq {p^2b^\ast_{0} }+ {2p^2(1+d)\psi}+ 2p(1+d)\psi $, and in general,
$b^\ast_i \leq \frac{2p(1+d)\psi}{1-p}$. This proves the claim.

Next, we observe that $b_\infty$ as defined in the statement of the lemma
is the fixed point of
\[
b =\frac{p(2a_{\infty}\psi+b)}{1+a_{\infty}}.
\]
Now, we define $\delta_i=b^\ast_i-b_{\infty}$ and show that $|\delta_i| \to 0$, which
yields the desired result. Note that
\begin{align*}
\delta_{i+1} &= b^\ast_{i+1} - b_\infty\\
             &=\frac{p(2a^\ast_i\psi+b^\ast_i)}{1+a^\ast_i} - \frac{p(2a_\infty\psi+b_\infty)}{1+a_\infty}\\
             &=  2p\psi\left(\frac{a^\ast_i}{1+a^\ast_i}- \frac{a_\infty}{1+a_\infty}\right)+\frac{pb^\ast_i}{1+a^\ast_i} - \frac{pb^\ast_i}{1+a_\infty}\\
              & \ \ \ + \frac{pb^\ast_i}{1+a_\infty} - \frac{pb_\infty}{1+a_\infty}\\
              &= \Delta_i + \bar{p}\delta_i,
\end{align*}
where
\[
\Delta_i = 2p\psi\left(\frac{a^\ast_i}{1+a^\ast_i}-\frac{a_{\infty}}{1+a_{\infty}}\right) + pb^\ast_i\left(\frac{1}{1+a^\ast_i} - \frac{1}{1+a_\infty}\right)
\]
and $\bar{p} = \frac{p}{1+a_\infty} < 1$. From triangle inequality,
$|\delta_{i+1}| \leq |\Delta_i|+\bar{p}|\delta_i|$. Moreover, since $a_i \to a_\infty$
and $b_i, i\geq0$, are bounded, $\Delta_i \to 0$. Hence, for any $\epsilon > 0$,
there exits a $i_\epsilon$ such that for all $i \geq i_\epsilon$, $\Delta_i \leq \epsilon$.
Hence $|\delta_{i_\epsilon + 1}| \leq  \bar{p}|\delta_{i_\epsilon}| + \epsilon$,
$|\delta_{i_\epsilon + 2}| \leq  \bar{p}^2|\delta_{i_\epsilon}| + \bar{p}\epsilon +
\epsilon$. In general,
\[
|\delta_i| \leq  \bar{p}^{(i-i_\epsilon)}|\delta_{i_\epsilon}| + \frac{\epsilon}{1-\bar{p}}
\]
for all $i \geq i_\epsilon$. So, $\lim_{i \to \infty}|\delta_i| \leq \frac{\epsilon}{1-\bar{p}}$.
Since $\epsilon$ can be chosen arbitrarily close to $0$,
$\lim_{i \to \infty}|\delta_i| = 0$.\\
We have $b^\ast_0 = 0 < 2\psi$. Now, assuming
$b^\ast_i < 2\psi$ for some $i$,
\[
b^\ast_{i+1} = \frac{p(2a^\ast_i\psi + b^\ast_i)}{1+a^\ast_i}  < \frac{2p\psi(1+a^\ast_i)}{1+a^\ast_i}<2\psi.
\]
Hence, by induction, $b^\ast_i < 2\psi$ for all $i \geq 0$.
\subsection{Proof of Lemma~\ref{Quadraticlemma:monotonicity-xk}}
\label{Quadraticapp:monotonicity-xk}
$(a)$To prove $\frac{2(x+\psi)-b_k}{2(1+a_k)}>0$, it suffices to prove the claim for $x = 0$. From Lemma~\ref{Quadraticlemma:monotonicity-ak-bk}$(b)$, the claim holds.\\
$(b)$Since $a_0=1+d$ and $b_0=0$, we clearly see that $\frac{2(x+\psi)-b_0}{2(1+a_0)}<\psi, \forall x \in [0,\psi]$.
We inductively prove that $\frac{2(x+\psi)-b_k}{2(1+a_k)}<\psi,\forall k \geq 0$. Let the result hold for the $k$-stage problem,
\begin{equation}
\label{Quadraticeqn:pi-k-less-shi}
\frac{2(x+\psi)-b_k}{2(1+a_k)}<\psi, \forall x \in [0,\psi].
\end{equation}
We argue that
\begin{equation*}
\frac{2(x+\psi)-b_{k+1}}{2(1+a_{k+1})}<\psi, \forall x \in [0,\psi].
\end{equation*}
Since the left hand side is increasing in $x$, it suffices to show
that
\begin{align*}
\frac{4\psi-b_{k+1}}{2(1+a_{k+1})} &< \psi \\
\text{or, } 2\psi a_{k+1}+b_{k+1} &> 2\psi.
\end{align*}
Using~\eqref{Quadraticeqn:ak-star} and ~\eqref{Quadraticeqn:bk-star},
\begin{align*}
\lefteqn{2\psi a_{k+1}+b_{k+1}} \\
& = 2\psi \left(1+d-\frac{p}{1+a_k}\right)+\frac{p(2a_k\psi+b_k)}{(1+a_k)}\\
& = 2\psi(1+d)+\frac{(b_kp+2pa_k\psi-2p\psi)}{(1+a_k)}\\
& = 2\psi(1+d)+p\frac{2\psi a_k+b_k-2\psi}{(1+a_k)}\\
& > 2\psi,
\end{align*}
where the last inequality is obtained by setting $x = \psi$ in~\eqref{Quadraticeqn:pi-k-less-shi}.
This completes the induction step.
\subsection{Proof of Theorem~\ref{Quadraticthm:optimal-policy}}
\label{Quadraticapp:optimal-policy}
Let us first recall the notions of $k$-stage problems and $k$-stage
optimal cost functions $J_k$. For all $k \geq 0$, we will express
$J_k$ as
\begin{equation}
\label{Quadraticeqn:jk-general-form}
J_{k}(x)=\min_{u \in [0,\psi]}\left\{(\psi-u+x)^2+dx^2+a_k u^2+b_k u +c_k\right\}.\hspace{-0.1in}
\end{equation}
Comparing with~\eqref{Quadraticeqn:J0}, $a_{0}=a^\ast_{0},b_{0}=b^\ast_{0},c_{0}=0$.

Considering the form of $J_k$ in~\eqref{Quadraticeqn:jk-general-form},
the optimal policy for the $k$-stage problem
\begin{equation}
\label{pi-k-bounded}
\pi_{k}(x) =\min\left\{\max\left\{\frac{2(x+\psi)-b_k}{2(1+a_k)},0\right\},\psi\right\}.
\end{equation}
Using Lemma~\ref{Quadraticlemma:monotonicity-xk} for $k = 0$, ~\eqref{pi-k-bounded} can be written as
\[
\pi_{0}(x) =\frac{2(x+\psi)-b^\ast_{0}}{2(1+a^\ast_{0})},
\]
and hence
\begin{align*}
J_0(x)=(\frac{2a^\ast_{0}(x+\psi)+b^\ast_{0}}{2(1+a^\ast_{0})})^2+dx^2 +a^\ast_{0}(\frac{2(x+\psi)-b^\ast_{0}}{2(1+a^\ast_{0})})^2\\
\ \ \ \ \ \ \ \ + b^\ast_{0}(\frac{2(x+\psi)-b^\ast_{0}}{2(1+a^\ast_{0})})+c^\ast_{0}
\end{align*}
where $c^\ast_{0}$ is a certain constant. Therefore, using~\eqref{Quadraticeqn:Jk}, $a_{1}=a^\ast_{1},b_{1}=b^\ast_{0}$. Therefore again using~\eqref{pi-k-bounded}, Lemma~\ref{Quadraticlemma:monotonicity-xk} for $k = 1$, it can be shown that  
\[
\pi_{1}(x) =\frac{2(x+\psi)-b^\ast_{1}}{2(1+a^\ast_{1})}.
\]
Continuing in the same fashion, we see that for all $k \ge 1$
\[
\pi_{k}(x) =\frac{2(x+\psi)-b^\ast_{k}}{2(1+a^\ast_{k})}.
\]
Further, from Lemma~\ref{Quadraticlemma:monotonicity-xk} and {\it $a^\ast_{k} \to a_\infty,b^\ast_k \to b_{\infty}$ as $k \to \infty$} it can be observed that
\[
\pi^\ast(x) =\frac{2(x+\psi)-b_\infty}{2(1+a_\infty)}.
\]
%\section{}
\subsection{Proof of Lemma~\ref{Quadraticlemma:nash-properties-a-b}}
The tagged user's optimal costs in the $k$-stage problems
\begin{equation}
    \label{Quadraticeqn:nash-c-0-cost}
C_0(x)=\min_{u \in [0,\psi]}\{(\psi-u)(\psi-u+x)+u^2(1+d)\}
\end{equation}
and for all $k \ge 1$,
\begin{align}
    \label{Quadraticeqn:nash-c-k-cost}
C_k(x)=\min_{u \in [0,\psi]}\left\{(\psi-u)(\psi-u+x)+du^2\right.\nonumber \\
\left.+u(u
+p(\psi-\pi'_{k-1}(u)))\right\}.
\end{align}
Recall that $\pi'_k(\cdot)$ denote the corresponding optimal policies. As before 
$lim_k \to \infty C_k(\cdot) = C(\cdot)$ and $\lim_k \to \infty  \pi'_k(\cdot) =  \pi'(\cdot)$.
\label{Quadraticappendix:nash-properties-a-b}
$(a)$ Notice that the mapping $a \mapsto \frac{1}{4+2d-2pa}$ is monotonically increasing.
Further, $a'_0 > a'_{-1}$. Therefore the sequence $a'_k, k\geq -1$ is monotonically increasing.
%It is also nonnegative, and so, lower bounded.
Hence it converges to $a'_\infty$, the smallest fixed point
of $a = \frac{1}{4+2d-2pa}$. 

Using $p \in [0,1]$ and the definition of $a'_{\infty}$ the following holds
\[
a'_{\infty}<\frac{1+\frac{d}{2}}{p}.
\]
$(b)$ As $\frac{4+2d}{2p}>\frac{1+\frac{d}{2}}{p}$, observe that $a'_{\infty} < \frac{4+2d}{2p}$. Hence $4+2d-2pa'_k$ is decreasing but strictly positive. 
We now show that $b'_i, i \geq 0$ are bounded. Towards this, observe that
$b'_i \leq pb'_{i-1} + (2-p)\psi $ for all $i \geq 1$. In particular,
$b'_1 \leq pb'_0 + (2-p)\psi $, $b'_2 \leq p^2 b'_0 + p((2-p)\psi) +(2-p)\psi$, and in general,
$b'_i \leq b'_0 + \frac{(2-p)\psi}{1-p}$. This proves the claim.

Next, we observe that $b_\infty$ as defined in the statement of the lemma
is the fixed point of
\[
    b=\frac{(2-p)\psi+pb}{4+2d-2pa'_{\infty}}.
    \]
Now, we define $\delta_i=b^\ast_i-b_{\infty}$ and show that $|\delta_i| \to 0$, which
yields the desired result. Note that
\begin{align*}
\delta_{i+1} &= b'_{i+1} - b'_\infty\\
             &=\frac{(2-p)\psi+pb'_i}{4+2d-2pa'_{i}} - \frac{(2-p)\psi+pb' _{\infty}}{4+2d-2pa'_{\infty}}\\
             &= ((2-p)\psi+pb'_i)(\frac{1}{4+2d-2pa'_{i}}-\frac{1}{4+2d-2pa'_{\infty}})\\ 
              & \ \ \ +\frac{p}{4+2d-2pa'_{\infty}}(b'_i-b'_{\infty})\\
              &= \Delta_i + \bar{p}\delta_i,
\end{align*}
where
\[
\Delta_i = ((2-p)\psi+pb'_i)(\frac{1}{4+2d-2pa'_{i}}-\frac{1}{4+2d-2pa'_{\infty}})
\]
and $\bar{p} = \frac{p}{4+2d-2pa'_{\infty}} < 1$. From triangle inequality,
$|\delta_{i+1}| \leq |\Delta_i|+\bar{p}|\delta_i|$. Moreover, since $a'_i \to a'_\infty$
and $b'_i, i\geq 0$, are bounded, $\Delta_i \to 0$. Hence, for any $\epsilon > 0$,
there exits a $i_\epsilon$ such that for all $i \geq i_\epsilon$, $\Delta_i \leq \epsilon$.
Hence $|\delta_{i_\epsilon + 1}| \leq  \bar{p}|\delta_{i_\epsilon}| + \epsilon$,
$|\delta_{i_\epsilon + 2}| \leq  \bar{p}^2|\delta_{i_\epsilon}| + \bar{p}\epsilon +
\epsilon$. In general,
\[
|\delta_i| \leq  \bar{p}^{(i-i_\epsilon)}|\delta_{i_\epsilon}| + \frac{\epsilon}{1-\bar{p}}
\]
for all $i \geq i_\epsilon$. So, $\lim_{i \to \infty}|\delta_i| \leq \frac{\epsilon}{1-\bar{p}}$.
Since $\epsilon$ can be chosen arbitrarily close to $0$,
$\lim_{i \to \infty}|\delta_i| = 0$.   
\subsection{Proof of Lemma~\ref{Quadraticlemma:quad-game-no-caps}}
\label{Quadraticapp:quad-game-no-caps}
To prove $a'_kx+b'_k>0,\forall x \in [0,\psi],\forall k \ge 0$, 
it suffices to prove for $x=0$. It in turn implies $b'_k \ge0$. From~\eqref{Quadraticeqn:b-case2-nash} it can be seen that 
$b'_{-1}>0$. Also from Lemma~\ref{Quadraticlemma:nash-properties-a-b} we can observe that $a_i \le a'_\infty<\frac{1+\frac{d}{2}}{p},\forall i\ge 0$. Therefore, the following holds true
\begin{align*}
    a'_i<\frac{4+2d}{2p}.
\end{align*}
Hence,
\[
    4+2d-2pa'_i>0.
\]
Also, $2-p \ge 0$, thus $b'_i>0,\forall i \ge 0$. Hence,
\[
a'_kx+b'_k>0,\forall x \in [0,\psi],\forall k \ge 0.
\]
To prove $a'_kx+b'_k<\psi,\forall x \in [0,\psi]$ it suffices to prove $a'_k\psi+b'_k<\psi$. From~\eqref{Quadraticeqn:a-case2-nash},\eqref{Quadraticeqn:b-case2-nash} it can  be verified that 
\[
a'_0\psi+b'_0<\psi.
\]
We inductively prove that $a'_k\psi+b'_k<\psi,\forall k \ge 0$. Let the following result hold
\begin{equation}
\label{Quadraticeqn:induction-step-game-quad}
a'_k\psi+b'_k<\psi.
\end{equation}
We argue that 
\[
a'_{k+1}\psi+b'_{k+1}<\psi.
\]
Using~\eqref{Quadraticeqn:a-case2-nash} and~\eqref{Quadraticeqn:b-case2-nash} we have 
\begin{align*}
    a'_{k+1}\psi+b'_{k+1}=\frac{\psi+(2-p)\psi+pb'_k}{2(2+d-a'_kp)}
    \end{align*}
    Hence, it suffices to prove the following
    \begin{align*}
   \frac{\psi+(2-p)\psi+pb'_k}{2(2+d-a'_kp)}&<\psi\\
   2a'_kp\psi+pb'_k&<(1+p+2d)\psi
    \end{align*}
Using~\eqref{Quadraticeqn:induction-step-game-quad} it is enough to show the following
    \begin{align*}
      a'_kp\psi+p\psi&<(1+p+2d)\psi\\
      a'_k &<\frac{1+2d}{p}
    \end{align*}
    Last inequality holds true from Lemma~\ref{Quadraticlemma:nash-properties-a-b}$(a)$. This completes the induction step. Hence the lemma follows.
    \subsection{Proof of Theorem~\ref{Quadraticthm:quad-nash-equilibrium}}
 \label{Quadraticappendix:quad-nash-equilibrium}
 Let us first recall the notion of $k$-stage problems and the corresponding optimal strategies. For all $k \ge 0$, we will express $\pi'_k(\cdot)$ as
 \[
 \pi'_k(x)=a_kx+b_k.
 \]
 Recall the functions $C_k(\cdot), k \ge 0$~(see~\eqref{Quadraticeqn:nash-c-0-cost}-\eqref{Quadraticeqn:nash-c-k-cost}); for $k \geq 0$,
 \begin{align*}
 C_k(x)=\min_{u \in [0,\psi]}\left\{(\psi-u)(\psi-u+x)+u(u(1+d) +p(\psi-a_{k-1}u - b_{k-1}))\right\}
 \end{align*}
 From~\cite[Chapter~2, Proposition~1.2(b)]{b21}, $C_k(\cdot)$s converge to the optimal cost function $C(\cdot)$ and $\pi'_k(\cdot)$ converge to $\pi'(\cdot)$ irrespective of the initial function $C_0(x)$ in the value iteration. Now, we analyze value iteration starting with $a_{-1} = a'_{-1}$ and $b_{-1} = b'_{-1}$. Considering the form of $C_k(x)$,
 \[
 \pi'_k(x)=\min\left\{\max\left\{\frac{x+(2-p)\psi+pb'_{k-1}}{4+2d-2pa'_{k-1}},0\right\},\psi\right\}.
 \]
 for all $k \ge 0$. From Lemma~\ref{Quadraticlemma:quad-game-no-caps}, it can be seen that $0<\pi'_0<\psi$. Hence, 
 \[
 \pi'_0(x)=\frac{x+(2-p)\psi+pb'_{-1}}{4+2d-2pa'_{-1}}.
  \]
Using~\eqref{Quadraticeqn:a-case2-nash} and~\eqref{Quadraticeqn:b-case2-nash}, it can be seen that $a_0=a'_0,b_0=b'_0$. Hence,
 \[
 \pi'_1(x)=\min\left\{\max\left\{\frac{x+(2-p)\psi+pb'_{0}}{4+2d-2pa'_{0}},0\right\},\psi\right\}.
  \]
Thus again using Lemma~\ref{Quadraticlemma:quad-game-no-caps}, it can be seen that $0<\pi'_1<\psi$.
\[
 \pi'_1(x)=\frac{x+(2-p)\psi+pb'_{0}}{4+2d-2pa'_{0}}.
\]
Similarly it can be argued that for all $k \ge 0$
\[
 \pi'_k(x)=\frac{x+(2-p)\psi+pb'_{k-1}}{4+2d-2pa'_{k-1}}.
\]
From Lemma~\ref{Quadraticlemma:nash-properties-a-b} as $\{a'_k\},\{b'_k\}$ converge to $a'_\infty,b'_\infty$ respectively, optimal policy $\pi'(x)$ can be written as
\[
 \pi'(x)=\frac{x+(2-p)\psi+pb'_{\infty}}{4+2d-2pa'_{\infty}}=a'_{\infty}x+b'_{\infty}.
\]
\color{black}
\subsection{Proof of Lemma~\ref{Lemma:Quadratic-newLemma}}
\label{Appendix:newLemma}
The following hold almost surely
\begin{align}
\label{eqn:pi-bound}
    |\pi^\ast(x)-\tilde{\pi}(x)|&=\Bigg | \frac{x+\psi-\frac{b(p)}{2}}{1+a(p)}-\frac{x+\psi-\frac{b(\tilde{p})}{2}}{1+a(\tilde{p})}\Bigg |\nonumber \\
    & =\Bigg |(x+\psi)\Bigg(\frac{1}{1+a(p)}-\frac{1}{1+a(\tilde{p})} \Bigg)+\Bigg(\frac{b(\tilde{p})(1+a(p))-b(p)(1+a(\tilde{p}))}{2(1+a(p))(1+a(\tilde{p}))}\Bigg)\Bigg |\nonumber \\
    &\le \Bigg |2\psi \Bigg(\frac{a(\tilde{p})-a(p)}{(1+a(p))(1+a(\tilde{p}))} \Bigg)\Bigg |+\Bigg |\Bigg(\frac{b(\tilde{p})(1+a(p))-b(p)(1+a(\tilde{p}))}{2(1+a(p))(1+a(\tilde{p}))}\Bigg)\Bigg |.
\end{align}
We now deal with the two terms separately. Let us first start bounding the first term. Let us define the first derivative of $a(\cdot)$ to be $a'(\cdot)$, which can be defined as 
\[
|a'(p)|=\frac{1}{\sqrt{d^2+4(1+d-p)}} \le \frac{1}{\sqrt{d^2+4d}}
\]
As the function $a(\cdot)$ is continuously differentiable and bounded the following would hold almost surely
\[
|a(\tilde{p})-a(p)| \le \frac{1}{\sqrt{d^2+4d}} |\tilde{p}-p|
\]
By definition of $a(p)$, the following holds
\begin{align*}
    \Bigg(\frac{1}{(1+a(p))(1+a(\tilde{p}))} \Bigg) \le \frac{4}{(2+d+\sqrt{d^2+4d})^2}
\end{align*}
Finally, almost surely
\begin{align}
\label{eqn:equationone}
    \Bigg |2\psi \Bigg(\frac{a(\tilde{p})-a(p)}{(1+a(p))(1+a(\tilde{p}))} \Bigg)\Bigg | \le \frac{8\psi}{(2+d+\sqrt{d^2+4d})^2\sqrt{d^2+4d}}  |\tilde{p}-p|
\end{align}
Let us now look at the second term in~\eqref{eqn:pi-bound}. Let us define the first derivative of $b(p)$ to be $b'(p)$, it can be written as follows
\[
b'(p)=2\psi \frac{a(p)+a(p)^2+pa'(p)-p^2a'(p)}{(1+a(p)-p)^2}.
\]
In the following we bound the first derivative of $b(p)$, using the facts
\begin{enumerate}
    \item $a(p) \le 1+d$
    \item $\frac{1}{(1+a(p)-p)^2} \le \frac{4}{(d+\sqrt{d^2+4d})^2}$
    \item $a'(p) \le \frac{1}{\sqrt{d^2+4d}}$
\end{enumerate}
Using the above facts the following holds
\begin{equation}
\label{eqn:b-derivative}
b'(p) \le 8\psi \Bigg( \frac{(1+d)(2+d)+\frac{1}{\sqrt{d^2+4d}}}{(d+\sqrt{d^2+4d})^2} \Bigg).
\end{equation}
As the function $b(\cdot)$ is continuously differentiable and bounded the following holds, almost surely
\begin{equation}
    \label{eqn:b-derivative-bounds}
    |b(\tilde{p})-b(p)|\le 8\psi \Bigg( \frac{(1+d)(2+d)+\frac{1}{\sqrt{d^2+4d}}}{(d+\sqrt{d^2+4d})^2} \Bigg) |\tilde{p}-p|
\end{equation}
We next bound $b(\tilde{p})a(p)-b(p)a(\tilde{p})$ on the above using the facts $b(p) \le 2\psi, a(p) \le 1+d$. The following holds almost surely
\begin{align*}
\label{eqn:b-a-bound}
  b(\tilde{p})a(p)-b(p)a(\tilde{p}) &= b(\tilde{p})a(p)- b(\tilde{p})a(\tilde{p})+b(\tilde{p})a(\tilde{p})-b(p)a(\tilde{p})\nonumber \\
  &=b(\tilde{p})\{a(p)-a(\tilde{p})\}+a(\tilde{p})\{b(\tilde{p})-b(p)\}\nonumber \\
  & \le 2\psi \{a(p)-a(\tilde{p})\}+(1+d)\{b(\tilde{p})-b(p)\}\nonumber
  \end{align*}
  \begin{align}
 |b(\tilde{p})a(p)-b(p)a(\tilde{p})| & \le \Bigg(\frac{2\psi}{\sqrt{d^2+4d}}+8\psi (1+d)\Bigg( \frac{(1+d)(2+d)+\frac{1}{\sqrt{d^2+4d}}}{(d+\sqrt{d^2+4d})^2} \Bigg)\Bigg) |\tilde{p}-p|
\end{align}
Let us come back to the second term, almost surely
\begin{align}
\label{eqn:second-term}
\Bigg |\Bigg(\frac{b(\tilde{p})(1+a(p))-b(p)(1+a(\tilde{p}))}{2(1+a(p))(1+a(\tilde{p}))}\Bigg)\Bigg | & \le \Bigg |\frac{b(\tilde{p})-b(p)}{2(1+a(p))(1+a(\tilde{p}))} \Bigg| + \Bigg | \frac{b(\tilde{p})a(p)-b(p)a(\tilde{p})}{2(1+a(p))(1+a(\tilde{p}))}\Bigg| \nonumber \\
& \le 4\psi \frac{4 (2+d) \Bigg(   \frac{(1+d)(2+d)+\frac{1}{\sqrt{d^2+4d}}}{(d+\sqrt{d^2+4d})^2} \Bigg) + \Bigg(\frac{1}{\sqrt{d^2+4d}}\Bigg)}{{(d+\sqrt{d^2+4d})^2}} |\tilde{p}-p|
\end{align}
Second inequality follows from~\eqref{eqn:b-a-bound},~\eqref{eqn:b-derivative-bounds} and Fact-2. Using~\eqref{eqn:equationone} and~\eqref{eqn:second-term}, almost surely
\[
|\pi^\ast(x)-\tilde{\pi}(x)| \le  \Bigg[  \frac{8\psi}{(2+d+\sqrt{d^2+4d})^2\sqrt{d^2+4d}}+ 4\psi \frac{4 (2+d) \Bigg(   \frac{(1+d)(2+d)+\frac{1}{\sqrt{d^2+4d}}}{(d+\sqrt{d^2+4d})^2} \Bigg) + \Bigg(\frac{1}{\sqrt{d^2+4d}}\Bigg)}{{(d+\sqrt{d^2+4d})^2}} \Bigg] |\tilde{p}-p|
\]
\subsection{Proof of Lemma~\ref{lemma:result-a}}
\label{Appendix:result-a}
We first fix $\epsilon,h$ thereby we obtain $\tilde{n}$. Note that $\tilde{\pi}^k(x)$ is a random variable that depends on $X_i,i\le \tilde{n}$.
Let us first begin with $|\tilde{\pi}(x_1)-\pi^\ast(x_2)|$, almost surely the following holds
\begin{align}
    |\tilde{\pi}(x_1)-\pi^\ast(x_2)| &\le |\tilde{\pi}(x_1)-\pi^\ast(x_1)|+|{\pi}^\ast(x_1)-\pi^\ast(x_2)|\nonumber \\
    & \le K |\tilde{p}-p|+\frac{|x_1-x_2|}{1+a(p)}
    \label{eqn:pi-basic-inequality}
\end{align}
Recollect the definition of $a(p)$
\begin{align}
   a(p)&=\frac{d+\sqrt{d^2+4(1+d-p)}}{2}\nonumber\\ 
   1+a(p)&=\frac{d+2+\sqrt{d^2+4(1+d-p)}}{2}\nonumber\\ 
   1+a(p)&\ge \frac{2+d+d}{2}\nonumber\\
   \frac{1}{1+a(p)} &\le \frac{1}{1+d}.
   \label{eqn:1+d}
\end{align}
Using~\eqref{eqn:1+d} in~\eqref{eqn:pi-basic-inequality} we obtain the following holds almost surely
\[
|\tilde{\pi}(x_1)-\pi^\ast(x_2)| \le K.|\tilde{p}-p|+\frac{|x_1-x_2|}{1+d}
\]
Hence the lemma holds.
\subsection{Proof of Lemma~\ref{lemma:result-b}}
\label{appendix:result-b}
We first fix $\epsilon,h$ thereby we obtain $\tilde{n}$. Note that  $\tilde{p}=\frac{1}{\tilde{n}}\sum_{i=1}^{\tilde{n}}X_i$. Also,  $\tilde{\pi}^k(x)$ is a random variable that depends on $X_i,i\le \tilde{n}$.
Result holds for $n=0$ from Lemma~\ref{lemma:result-a} with $x_1=x_2=x$.

Let the lemma for $n=k$, i.e., the following holds almost surely
\begin{equation}
\label{eqn:inequality-k}
    |\tilde{\pi}^k(x)-{\pi^\ast}^k(x)|\le \sum_{i=0}^k z^i K|\tilde{p}-p|
\end{equation}
To complete the induction step we now consider $n=k+1$. Using Lemma~\ref{lemma:result-a} with $x_1=\tilde{\pi}^k(x),x_2={\pi^\ast}^k(x)$ the following holds almost surely
\begin{align*}
|\tilde{\pi}^{k+1}(x)-{\pi^\ast}^{k+1}(x)|&\le K|\tilde{p}-p|+z|\tilde{\pi}^k(x)-{\pi^\ast}^k(x)|,\\
&\le K|\tilde{p}-p|+z\sum_{i=0}^k z^i K|\tilde{p}-p|,\\
&= \sum_{i=0}^{k+1} z^i K|\tilde{p}-p|.
\end{align*}
Hence the lemma holds.
\subsection{Proof of Lemma~\ref{lemma:result-c}}
\label{Appendix:result-c}
We first fix $\epsilon,h$. Note that $\tilde{p}=\frac{1}{\tilde{n}}\sum_{i=1}^{\tilde{n}}X_i$. Note that  $\tilde{J}_k(x)$ is a random variable that depends on $X_i,i\le \tilde{n}$. Using~\eqref{Quadraticeqn:J0-pi} and~\eqref{Quadraticeqn:J0-pi-tilda} the following holds almost surely as $x,\tilde{\pi}(\cdot),\pi^\ast(\cdot) \in [0,\psi]$.
\begin{align}
\label{eqn:j0-inequality}
    \tilde{J}_0(x_1)-J_0(x_2) \le 2\psi \{(d+2)|x_1-x_2|+(3+d)|\tilde{\pi}(x_1)-\pi^\ast(x_2)|\}.
\end{align}
Using the above inequality with $x_1=\tilde{\pi}^y(x),x_2={\pi^\ast}^y(x)$, we obtain the following almost surely
\[
\tilde{J}_0(\tilde{\pi}^y(x))-J_0({\pi^\ast}^y(x)) \le 2\psi \{(d+2)|\tilde{\pi}^y(x)-{\pi^\ast}^y(x)|+(3+d)|\tilde{\pi}^{y+1}(x))-{\pi^\ast}^{y+1}(x)|\}.
\]
Using Lemma~\ref{lemma:result-b} with $n=y, n=y+1$ the following holds almost surely
\begin{align*}
\tilde{J}_0(\tilde{\pi}^y(x))-J_0({\pi^\ast}^y(x)) &\le 2\psi \{(d+2)\sum_{i=0}^y z^i+(3+d)\sum_{i=0}^{y+1} z^i\}K |\tilde{p}-p| \\
&\le \frac{2K\psi (5+2d)}{1-z} |\tilde{p}-p|,
\end{align*}
where the last inequality holds as $z<1$. Hence the result holds.
\subsection{Proof of Lemma~\ref{lemma:result-II}}
\label{Appendix:result-II}
We first fix $\epsilon,h$. Note that $\tilde{p}=\frac{1}{\tilde{n}}\sum_{i=1}^{\tilde{n}}X_i$. Note that  $\tilde{J}_k(x)$ is a random variable that depends on $X_i,i\le \tilde{n}$. Using~\eqref{Quadraticeqn:Jk-pi} and~\eqref{Quadraticeqn:Jk-pi-tilda} the following holds almost surely
\begin{align*}
\tilde{J}_{k-n}(\tilde{\pi}^{n-1}(x))-{J}_{k-n}({\pi^\ast}^{n-1}(x)) & \le \frac{A}{K}|\tilde{\pi}^n(x)-{\pi^\ast}^n(x)|+\frac{B}{K}|\tilde{\pi}^{n-1}(x)-{\pi^\ast}^{n-1}(x)|\\
&~~~~~ +p\{\tilde{J}_{k-n-1}(\tilde{\pi}^{n-1}(x))-{J}_{k-n-1}({\pi^\ast}^{n-1}(x))\}.
\end{align*}
Using Lemma~\ref{lemma:result-b}, in the above inequality the following almost surely
\begin{align}
\tilde{J}_{k-n}(\tilde{\pi}^{n-1}(x))-{J}_{k-n}({\pi^\ast}^{n-1}(x)) & \le \frac{A}{K}\sum_{i=0}^n z^i K|\tilde{p}-p|+\frac{B}{K}\sum_{i=0}^{n-1} z^i K|\tilde{p}-p|\nonumber\\
&~~~~~ +p\{\tilde{J}_{k-n-1}(\tilde{\pi}^{n-1}(x))-{J}_{k-n-1}({\pi^\ast}^{n-1}(x))\}\nonumber\\
& = [(A+B)\sum_{i=0}^{n-1} z^i+Az^n]|\tilde{p}-p|\nonumber\\
&~~~~~ +p\{\tilde{J}_{k-n-1}(\tilde{\pi}^{n-1}(x))-{J}_{k-n-1}({\pi^\ast}^{n-1}(x))\}.
\label{eqn:base}
\end{align}
Using Lemma~\ref{lemma:result-c} with $y=k-1$, the following holds
\begin{equation}
\label{eqn:j0-resultII}
    \tilde{J}_{0}(\tilde{\pi}^{k-1}(x))-{J}_{0}({\pi^\ast}^{k-1}(x)) \le K'|\tilde{p}-p| 
\end{equation}
Using~\eqref{eqn:base} with $n=k-1$ and~\eqref{eqn:j0-resultII} the following holds 
\begin{equation}
\label{eqn:j0-resultII}
    \tilde{J}_{1}(\tilde{\pi}^{k-2}(x))-{J}_{1}({\pi^\ast}^{k-2}(x)) \le [(A+B)\sum_{i=0}^{k-2}z^i+Az^{k-1}]+pK']|\tilde{p}-p| 
\end{equation}
Similarly, iteratively using~\eqref{eqn:base} with $n=\{k-2,\dots,n\}$ and reusing those results the lemma holds.
\subsection{Proof of Theorem~\ref{theorem:bound-unknown}}
\label{Appendix:bound-unknown}
We first fix $\epsilon,h$, thereby we obtain a $\tilde{n}=-\frac{1}{2\epsilon^2}\log{\frac{1-h}{2}}$ such that $\forall n \ge \tilde{n}$
\begin{equation*}
    P\Big(|\frac{1}{n}\sum_{i=1}^nX_i-p|\le \epsilon \Big) \ge h.
\end{equation*} 
Note that $\tilde{p}=\frac{1}{\tilde{n}}\sum_{i=1}^{\tilde{n}}X_i$. Note that  $\tilde{J}_k(x)$ is a random variable that depends on $X_i,i\le tilde{n}$.
From~\eqref{Quadraticeqn:Jk-pi-tilda} and~\eqref{Quadraticeqn:Jk-pi}, the following holds almost surely
\begin{align}
    \tilde{J}_k(x)-J_k(x)&=(\psi+x-\tilde{\pi}(x))^2-(\psi+x-{\pi^\ast}(x))^2+(1-p)(1+d)(\tilde{\pi}(x)^2-\pi^\ast(x)^2)\nonumber\\
    &~~~~+p\{\tilde{J}_{k-1}(\tilde{\pi}(x))-J_{k-1}(\pi^\ast(x))\}\nonumber\\
    & \le 2\psi\{(1-p)(1+d)+2\}|\pi^\ast(x)-\tilde{\pi}(x)|+p\{\tilde{J}_{k-1}(\tilde{\pi}(x))-J_{k-1}(\pi^\ast(x))\}\nonumber\\
    & \le 2\psi\{(1-p)(1+d)+2\}K|\tilde{p}-p|+p\{\tilde{J}_{k-1}(\tilde{\pi}(x))-J_{k-1}(\pi^\ast(x))\}\nonumber\\
    &=A|\tilde{p}-p|+p\{\tilde{J}_{k-1}(\tilde{\pi}(x))-J_{k-1}(\pi^\ast(x))\}
    \label{eqn:0-stage-bound}
    \end{align}
The first inequality follows as $x \in [0,\psi]$ and $\pi^\ast(\cdot),\tilde{\pi}(\cdot) \in [0,\psi]$. Second inequality follows from Lemma~\ref{Lemma:Quadratic-newLemma}. 

Using Lemma~\ref{lemma:result-II} with $n=1$ in~\eqref{eqn:0-stage-bound} we obtain the following almost surely
\[
\tilde{J}_k(x)-J_k(x) \le \Bigg(A+p\Big[(A+B)\sum_{j=0}^{k-1}p^j \sum_{i=0}^{j}z^i+Az^n \sum_{j=0}^{k-1} (pz)^j+p^{k-1}K'\Big]\Bigg)|\tilde{p}-p|
\]
As $k \to \infty$, the above inequality boils down to
\[
\tilde{J}(x)-J(x) \le \Bigg[A+\frac{A+B}{(1-z)(1-p)}+\frac{Az}{1-pz}\Bigg]|\tilde{p}-p|,
\]
almost surely as $z<1$ and $K'$ is finite. Note that $p$ is an unknown parameter, however after $\forall n \ge \tilde{n}$, we know with atleast probability $h$, 
\[
|\tilde{p}-p| \le \epsilon
\]
For fixed $\epsilon,h$, there exists a $\tilde{n}=-\frac{1}{2\epsilon^2}\log{\frac{1-h}{2}}$  such that $\forall n \ge \tilde{n}$, with a probability of atleast $h$ the following holds
    \[
    \tilde{J}(x)-J(x) \le \Bigg[A+\frac{A+B}{(1-z)(1-\tilde{p}-\epsilon)}+\frac{Az}{1-(\tilde{p}+\epsilon)z}\Bigg]\epsilon
    \]
Hence the theorem holds.

    \color{black}

\end{appendix}

%\bibliographystyle{IEEEtran}
%\bibliography{references}
\end{document}